\documentclass{svproc}

\usepackage{graphicx}%
\usepackage{multirow}%
\usepackage{amsmath,amssymb,amsfonts}%
\usepackage{mathrsfs}%
\usepackage[title]{appendix}%
\usepackage{xcolor}%
\usepackage{textcomp}%
\usepackage{manyfoot}%
\usepackage{booktabs}%
\usepackage{algorithm}%
\usepackage{algorithmicx}%
\usepackage{algpseudocode}%
\usepackage{listings}%
\usepackage{url}

\usepackage[utf8]{inputenc}
\usepackage{xspace}
\usepackage{subfig}
\usepackage{mathtools}
\usepackage{hyperref}
\usepackage[font=small,labelfont=bf]{caption}

\usepackage[numbers]{natbib}

\newcommand{\w}{\mathbf{w}}

\newcommand{\head}[1]{\vspace{1.7mm}\noindent{{\bf #1.}}}

\newtheorem{dfn}{Definition}
\newtheorem{prop}{Proposition}

\usepackage{tikz}
\newcommand*\circled[1]{\tikz[baseline=(char.base)]{
            \node[shape=circle,draw,inner sep=0.3pt] (char) {#1};}}

\begin{document}
\mainmatter              % start of a contribution
\title{Generalized Densest Subgraph in Multiplex Networks}
\titlerunning{Generalized Densest Subgraph in Multiplex Networks}  % abbreviated title (for running head)
\author{Ali Behrouz\inst{1} \and Farnoosh Hashemi\inst{1}}
\authorrunning{Ali Behrouz and Farnoosh Hashemi} % abbreviated author list (for running head)

\institute{Cornell University, Ithaca, NY, USA\\
\email{\{ab2947, Sh2574\}@cornell.edu},}

\maketitle              % typeset the title of the contribution

\vspace{
-3ex
}
\begin{abstract}
Finding dense subgraphs of a large network is a fundamental problem in graph mining that has been studied extensively both for its theoretical richness and its many practical applications over the last five decades. However, most existing studies have focused on graphs with a single type of connection. In applications such as biological, social, and transportation networks, interactions between objects span multiple aspects, yielding multiplex graphs. Existing dense subgraph mining methods in multiplex graphs consider the same importance for different types of connections, while in real-world applications, one relation type can be noisy, insignificant, or irrelevant. Moreover, they are limited to the edge-density measure, unable to change the emphasis on larger/smaller degrees depending on the application. To this end, we define a new family of dense subgraph objectives, parametrized by two variables $p$ and $\beta$, that can (1) consider different importance weights for each relation type, and (2) change the emphasis on the larger/smaller degrees, depending on the application. Due to the NP-hardness of this problem, we first extend the FirmCore, $k$-core counterpart in multiplex graphs, to layer-weighted multiplex graphs, and based on it, we propose two polynomial-time approximation algorithms for the generalized densest subgraph problem, when $p \geq 1$ and the general case. Our experimental results show the importance of considering different weights for different relation types and the effectiveness and efficiency of our algorithms.

\end{abstract}

\section{Introduction}\label{sec:introduction}
Multiplex (ML) networks~\citep{main-ML} have become popular in various applications involving complex networks such as social, transportation, and biological networks. These networks involve interactions between objects that span different aspects. For instance, interactions between individuals can be categorized as social, family, or professional, and professional interactions can vary depending on the topic. ML networks allow nodes to have interactions in multiple relation types and represent the graph of each relation type as a layer in the network.

\noindent
Detecting Dense structures in a graph has become a key graph mining primitive with a wide range of applications~\cite{bio-dense, finance-dense, web-dense}. The common method for identifying dense subgraphs is to formulate an objective function (called density) that captures the density of each node set within a graph and then solve it via combinatorial optimization methods~\cite{densest_first, f-density, flow-densest}. While the problem of finding the densest subgraph in simple graphs is a well-studied problem in the literature and its recent advancements bring the problem close to being fully resolved~\cite{densest-survey}, extracting dense subgraphs from ML networks recently attracts attention~\cite{MLcore, FirmCore, FirmTruss}. Due to the complex interactions between nodes in ML networks, the definition of edge density is challenging. To this end, several studies~\cite{MLcore, densest-common-subgraph, FirmCore} introduced new density objective functions to capture complex dense subgraphs; however, in practice, it can be challenging to evaluate tradeoffs between density measures and decide which density is more useful. Accordingly, there is a lack of a unified framework that can generalize all the existing density measures to formalize the tradeoff between them.

\noindent
One of the main advantages of ML networks is their ability to provide complementary information by different relation types~\cite{FirmCore}. That is, some dense subgraphs can be missed if we only look at one relation type or the aggregated network~\cite{FirmCore}. However, taking advantage of this complementary information is challenging as in real-world applications, different relation types have different importance (e.g., some layers might be noisy/insignificant~\cite{MLcore, FirmCore, anomuly, CS-MLGCN}, or have different roles in the applications~\cite{flight, admire++, admire}). Existing dense subgraph models treat relation types equally, which means noisy/insignificant layers (or less important layers) are considered as important as other layers, causing suboptimal performance and missing some dense subgraphs (we support this claim in \S~\ref{sec:experiments}).

\noindent
To overcome the above challenges, we introduce a new family of density objectives in ML networks, $p$-mean multiplex densest subgraph ($p$-mean MDS), that: \circled{1} is able to handle different weights for layers, addressing different importance of relation types; \circled{2} given a parameter $p$, inspired by~\citet{p-mean}, it uses $p$-mean of node degrees in different layers. This design gives us the flexibility to emphasize smaller/larger degrees and allows us to uncover a hierarchy of dense subgraphs in the same ML graph; \circled{3} unifies the \emph{existing} definition of density in ML networks, which allows evaluating the tradeoffs between them. The multiplex $p$-mean density objective uses parameter $\beta$ to model the trade-off between high density and the cumulative importance of layers exhibiting the high density, and uses parameter $p$ to define $p$-mean of node degrees within a subgraph as a measure of high density (we formally define it in \S~\ref{sec:method}). Inspired by FirmCore strcture~\cite{FirmCore}, we further extend the concept of $k$-core to weighted layer ML networks and define weighted $(k. \lambda)$-FirmCore ($(k. \lambda)$-GFirmCore) as a maximal subgraph in which every node is connected to at least $k$ other nodes within that subgraph, in a set of layers with cumulative importance of at least $\lambda$. We discuss that given $\lambda$, weighted FirmCore has linear time decomposition in terms of the graph size, and can provide two tight approximation algorithms for the two cases of the $p$-mean MDS problem when \circled{i} $p \geq 1$ and \circled{ii} the general case.

\section{Related Work and Background}
Given the wide variety of applications for dense subgraph discovery~\cite{bio-dense, web-dense, finance-dense}, several variants of the densest subgraph problem with different objective functions have been designed~\cite{f-density, densest_first, p-mean, p-mean2}. Recently, \citet{p-mean} unifies most existing density objective functions and suggests using $p$-mean of node degrees within the subgraph as its density. In this case, when $p=1$, $p=-\infty$, and $p = 2$ we have the traditional densest subgraph problem, maximal $k$-core, and F-density~\cite{f-density}, respectively. Despite the usefulness of the family of $p$-mean density objectives, they are limited to simple graphs and their extension to ML networks is not straightforward.

\noindent
In ML networks, \citet{densest-common-subgraph} formulate the densest common subgraph problem and develop a linear-programming formulation. \citet{azimi-etal} propose a new definition of core, \textbf{k}-core, over ML graphs. \citet{MLcore} propose algorithms to find all possible \textbf{k}-cores, and generalized the formulation of \citet{densest-common-subgraph} by defining the density of a subgraph in ML networks as a real-valued function $\rho: 2^V \rightarrow \mathbb{R}^+$: 
\begin{equation}
    \rho(S) = \max_{\hat{L} \subseteq L} \min_{\ell \in \hat{L}} \frac{|E_\ell[S]|}{|S|} |\hat{L}|^\beta,
\end{equation}
\noindent
where $E_\ell[S]$ is the number of internal edges of $S$ in layer $\ell$, and $\beta \geq 0$ is a real number. They further propose a core-based $\frac{1}{2|L|^{\beta}}$-approximation algorithm. However, their algorithm takes exponential time in the number of layers, rendering it impractical for large networks (see\S~\ref{sec:experiments}). Recently, \citet{FirmCore} introduce FirmCore, a new family of dense subgraphs in ML network, as a maximal subgraph in which every node is connected to at least $k$ other nodes within that subgraph, in each of at least $\lambda$  individual layers.

Although the densest FirmCore approximates function $\rho(.)$, which its optimization is NP-hard~\cite{ml-core-journal},~with provable guarantee, it is limited to unweighted layer ML networks, missing some dense structures. Moreover, its approximation guarantee is limited to the objective function defined by \citet{MLcore},~and its performance~in~our~$p$-mean~MDS~is unexplored. For~additional~related work on the densest subgraph problem, we refer to the recent survey by \citet{densest-survey}.

\section{$p$-mean Multiplex Densest Subgraph}\label{sec:method}
We let $G = (V, E, L, \w)$ denote an ML graph, where $V$ is the set of nodes, $L$ is the set of layers, $E \subseteq V \times V \times L$ is the set of edges, and $\w(.): L \rightarrow \mathbb{R}^{\geq 0}$ is a function that assigns a weight to each layer. The set of neighbors of node $v \in V$ in layer $\ell \in L$ is denoted $N_\ell(v)$ and the degree of $v$ in layer $\ell$ is $\text{deg}_\ell (v) = |N_\ell(v)|$. For a set of nodes $H \subseteq V$, $G_\ell[H] = (H, E_\ell[H])$ shows the subgraph of $G$ induced by $H$ in layer $\ell$, and $\text{deg}^H_{\ell}(v)$ is the degree of $v$ in this subgraph. We sometimes use $G_\ell[V]$ and $E_\ell[V]$ as $G_\ell$ and $E_\ell$, respectively. 

\noindent
As discussed in \cite{MLcore}, the density in ML networks should be modeled as a trade-off between the high density and the number of layers exhibiting the high density. Here, we use this intuition and first use $p$-mean density to measure the density of the subgraph in each layer, i.e., 
\begin{equation}
    \Omega_\ell(S) = \left( \frac{1}{|S|} \sum_{u \in S} deg_\ell(u)^p \right)^{1/p},
\end{equation}
and then multiply it by the importance of the layer exhibiting this density:
\begin{equation}
     \Xi_\ell(S) = \Omega_\ell(S) \w(\ell).
\end{equation}
\noindent
Based on this definition of density we define the $p$-mean MDS problem as follows:

\begin{problem} [$p$-mean Multiplex Densest Subgraph]\label{prob:MDS}
Given an ML graph $G=(V,E,L, \w)$, real numbers $\beta \geq 0$ and $p \in \mathbb{R} \cup \{ +\infty, -\infty \}$, and a real-valued function $\rho : 2^V \rightarrow \mathbb{R}^+$ defined~as:
\begin{equation} \label{eq1}
\rho(S) = \max_{\hat{L} \subseteq L} \min_{\ell \in \hat{L}}   \Xi_\ell(S)
\left(\sum_{\ell' \in \hat{L}} \w(\ell')\right)^\beta,
\end{equation}
find a subset of vertices $S^* \subseteq V$ that maximizes $\rho$ function.
\end{problem}
\noindent
Note that given layer weights $\w(\ell)$, we aim to solve a max-min problem over~$\Xi_\ell(S)$. Also, given a layer $\ell$, maximizing the $\Xi_\ell(S)$ is equivalent to maximizing $\Omega_\ell(S)^p$~for $p > 0$ and minimizing $\Omega_\ell(S)^p$ for $p < 0$. Therefore, for the sake of simplicity,~in the following we aim to optimize (maximize or minimize) $\Omega_\ell(S)^p$. 
Following, we use $\Delta_\ell(S/\{u\}) = \Omega_\ell(S)^p - \Omega_\ell(S/\{u\})^p$, to denote the difference that~removing~a node $u$ can cause to the density of layer $\ell$. When $p = 1$ and $\w(.) = 1$, the $p$-mean MDS problem reduces to ML densest subgrapah problem~\cite{MLcore}.

\subsection{Generalized FirmCore Decomposition}
Next, inspired by the success FirmCore~\cite{FirmCore} in approximating the ML densest subgraph problem, we generalized it to layer-weighted ML networks and design an algorithm to find all existing FirmCores. In \S~\ref{sec:approx}, we use the generalized FirmCore to approximate Problem~\ref{prob:MDS}. 

There are two steps to generalize this concept: \circled{1} FirmCore treats all layers the same and consider the number of selected layers, accordingly. However, generalized FirmCore needs to consider the cumulative importance of selected layers, to take advantage of layer weights. \circled{2} In simple densest subgraph problem (i.e., $p = 1$), each node in a subgraph contributes the same to the denominator of the density function (i.e., subgraph size $|S|$), while each node's contribution to the numerator (i.e., number of edges) is as much as its degree. Traditionally, core structures attracts attention to approximate the densest subgraph as they provide lower bound for the minimum degree. However, in the $p$-mean density, the contribution of each node does not equal to its degree. As we discussed above, removing each node makes  $\Delta_\ell(S/\{u\}) = \Omega_\ell(S)^p - \Omega_\ell(S/\{u\})^p$ difference to the numerator of the $\Omega_\ell^p(S)$. Accordingly, in the general case $p \in \mathbb{R} \cup \{-\infty, \infty \}$, we want our generalized FirmCore to provide lower bound for the $\Delta_\ell(S/\{u\})$.

\begin{dfn}[Generalized FirmCore]\label{FirmCore}
Given an ML graph $G$, a non-negative \underline{real-value} threshold $\lambda$, an integer $k \geq 0$, and $p \in \mathbb{R} \cup \{ -\infty, +\infty \}$, the $(k, \lambda, p)$-GFirmCore of $G$ is a maximal subgraph $H = G[C_{k}] = (C_{k}, E[C_{k}], L)$ such that for each node $v\in C_k$ there are some layers with cumulative importance of at least $\lambda$ (i.e., $ \exists \{\ell_1, ..., \ell_s\} \subseteq L$ with $\sum_{i = 1}^s \w(\ell_i) \geq \lambda$) such that $\Delta_\ell(S/\{u\}) \geq k$, for $1\leq i\leq s$.
\end{dfn}

\begin{prop}
    When $p = 1$ and $\w(\ell_i) = 1$ for all $ \ell_i \in L$, $(k, \lambda, p)$-GFirmCore is equivalent to the $(k, \lambda)$-FirmCore~\cite{FirmCore}. 
\end{prop}

\begin{prop}[Hierarchical Structure]\label{prop:FirmCore-hierarchical}
Given a real-value threshold $\lambda$, an integer $k \geq 0$, and $p \in \mathbb{R} \cup \{-\infty, \infty \}$ the $(k + 1, \lambda, p)$-GFirmCore and $(k, \lambda + \epsilon, p)$-GFirmCore of~$G$ are subgraphs of its $(k, \lambda, p)$-GFirmCore for any $\epsilon \in \mathbb{R}^{+}$.
\end{prop}
\noindent
From now, to avoid confusion, when we refer to $(k, \lambda)$-GFirmCore, we assume that $\lambda$ is maximal. That is, for at least one vertex $u$ in $(k, \lambda)$-GFirmCore, there is a subset of layers with an exact summation of $\lambda$ in which $u$ has a degree not less than $k$. Next, we show that GFirmCore decomposition is strictly harder then the FirmCore decomposition, which is solvable in polynomial time, unless $P = NP$.

\begin{theorem}
    GFirmCore decomposition, which is finding all possible GFirmCores in an ML network, is NP-hard. 
\end{theorem}
\begin{proof}
    Here we provide the proof sketch for the sake of space constraint. Given a sequence of layer weights $w_1, w_2, \dots, w_{|L|}$,  the decision problem of whether there is a non-empty $(k, \lambda, p)$-GFirmCore can be simply reduced to the well-known NP-hard problem of the \emph{Subset Sum} over $w_1, w_2, \dots, w_{|L|}$, as its YES (resp. NO) instance means there is (resp. is not) a subset of $w_i$s with summation of~$\lambda$.
\end{proof}

\head{Algorithm}
Here, we design a polynomial-time algorithm that finds all $(k, \lambda, p)$-GFirmCores for given $\lambda$ and $p$. Given $\lambda$ and $p$, we define the GFirmCore index of a node $u$, Gcore$_\lambda(u)$, as the set of all $k \in \mathbb{N}$, such that $u$ is part of a $(k, \lambda, p)$-GFirmCore. For each node $u$ in subgraph $G[H]$, we consider a vector $\Psi(u)$ that its $\ell$-th element, $\Psi_\ell(u)$, shows $\Delta_\ell(H/\{u\})$'s in layer $\ell$. We further define Top-$\lambda(\Psi(u))$ as the maximum value of $k$ that there are some layers $\{\ell_1, \dots, \ell_t\}$ with a cumulative weight of at least $\lambda$ in which $\Delta_\ell(H/\{u\}) \geq k$. To calculate the Top-$\lambda(\Psi(u))$, we can simply sort the vector $\Psi(u)$ and check if the cumulative weights of layers in which $u$ has a $\Delta_\ell(H/\{u\})$ more than $k$ is $\geq \lambda$ or not. This process takes $\mathcal{O}(|L|\log |L|)$ time. It is easy to see that $u$ can be in at most $(k, \lambda, p)$-GFirmCore, where $k = $Top-$\lambda(\Psi(u))$. Accordingly, Algorithm \ref{alg:GFirmCore} processes the nodes in increasing order of Top$-\lambda(\Psi(u))$. It uses a vector $B$ of lists such that each element $i$ contains all nodes with Top$-\lambda(\Psi(u)) = i$. This technique allows us to keep vertices sorted throughout the algorithm and to update each element in $\mathcal{O}(1)$ time. Algorithm \ref{alg:GFirmCore} first initializes $B$ with Top$-\lambda(\Psi(u))$ and then starts processing $B$'s elements in increasing order. If a node $u$ is processed at iteration $k$, its Gcore$_\lambda$ is assigned to $k$ and removed from the graph. In order to remove a vertex from a graph, we need to update the degree of its neighbors in each layer, which leads to changing the Top$-\lambda(\Psi)$ of its neighbors and changing their bucket accordingly (lines 10-12). Note that it is simple to show that the above algorithm can find all $(k, \lambda, p)$-GFirmCores, given $\lambda$ and $p$. That is, at the end of $(k - 1)$-th iteration, each remaining nodes like $u$ has Top$-\lambda(\Psi(u)) \geq k$ as we removed all nodes with Top$-\lambda(\Psi)$ less than $k$ in the $(k-1)$-th iteration. 

\begin{algorithm}[t]
    \small
    \caption{Finding all $(k, \lambda, p)$-GFirmCores for a given $\lambda$}
    \label{alg:GFirmCore}
    \begin{algorithmic}[1]
        \Require{An ML graph $G = (V, E, L, \w)$, and a threshold $\lambda \in \mathbb{R}^{\geq 0}$}
        \Ensure{GFirmCore index Gcore$_\lambda(v)$ for each $v \in V$}
        \For{$v \in V$}
            \State $I[v] \leftarrow \text{Top-$\lambda$}(\Psi(v))$ 
            \State $B[I[v]] \leftarrow B[I[v]] \cup \{v\}$
        \EndFor
        \For{$k = 1, 2, \dots, |V|$}
            \While{$B[k] \neq \emptyset$}
                \State pick and remove $v$ from $B[k]$
                \State Gcore$_\lambda(v) \leftarrow k$, $N \leftarrow \emptyset$
                \For{$(v, u, \ell) \in E$ and $I[u] > k$}
                    \State update $\Psi_{\ell}(u)$ and remove $u$ from $B[I[u]]$
                    \State update $I[u]$  and $B[I[u]] \leftarrow B[I[u]] \cup \{u\}$
                \EndFor
                \State $V \leftarrow V \:\char`\\ \: \{v\}$
            \EndWhile
        \EndFor
    \end{algorithmic}
\end{algorithm}

\subsection{Approximation Algorithms}\label{sec:approx}
Algorithm~\ref{alg:ApproxGFirmCore} shows the pseudocode of the proposed approximation algorithm. Given a threshold $\alpha$, we first construct a candidate set for the value of $\lambda$. To this end, we consider the set of summations of all possible subsets of layer weights with size $1 \leq s \leq \alpha$, denoted as $\mathcal{M}$. Next, we use Algorithm~\ref{alg:GFirmCore} for each $\lambda \in \mathcal{M}$, and then report the densest GFirmCore as the approximate solution. In our experiments, we observe that always $\alpha = 10$ results in a good approximate solution. Given $p$, let  $S^*_{\text{SL}}$ be the $p$-mean densest subgraph among all single-layer densest subgraphs, and $\ell^*$  denote its layer. Let $C^*$ and $S^*$  denote our found approximation solution and the optimal solution, respectively. Finally, we use $\w^*$, $\w_{\min}$, and $\w_{\max}$ to refer to the summation of all layer weights, minimum weight, and maximum weight, respectively.

\begin{algorithm}[t]
    \small
    \caption{Approximation algorithm for the $p$-mean MDS}
    \label{alg:ApproxGFirmCore}
    \begin{algorithmic}[1]
        \Require{An ML graph $G = (V, E, L, \w)$,  a parameter $p \in \mathbb{R} \cup \{-\infty, \infty\}$}, and parameter $\alpha \in \{1, \dots, L \}$.
        \Ensure{Approximation solution to $p$-mean MDS.}
        \State $\mathcal{M} \leftarrow $ summations of all possible subsets of layer weights with size $1 \leq s \leq \alpha$;
        \For{$\lambda \in \mathcal{M}$}
        \State $\mathcal{Q}_{\lambda} \leftarrow $ find all $(k, \lambda, p)$-GFirmCore \Comment{Using Algorithm~\ref{alg:GFirmCore}}
        \State $\hat{C}_{\lambda} \leftarrow$ calculate the density and find the densest $(k, \lambda, p)$-GFirmCore $\in \mathcal{Q}_{\lambda} \:\: \rho()$.
        \EndFor
        \Return the densest subgraph among all $\hat{C}_{\lambda}$ for $\lambda \in \mathcal{M}$.
    \end{algorithmic}
\end{algorithm}

\begin{lemma}\label{density-firmcore}
Let $C$ be the $(k, \lambda, p)$-GFirmCore of $G$, we have:
\begin{align}\label{eq:density-firmcore0}
\rho(C) &\geq \frac{k^{1/p}}{{\w^{*}}^{1/p}} \times \underset{\tilde{L} \subseteq \hat{L}}{\max} \left\{ \left(\lambda - \sum_{i = 1}^{|\tilde{L}|} \w(\ell_i) \right)^{1/p}\hspace{-3ex} \times \underset{\ell \in \tilde{L}}{\max} ~\w(\ell)  \times \left(\sum_{\ell \in \tilde{L}} \w(\ell)\right)^\beta \right\}\\  \label{eq:density-firmcore}
&\geq \frac{k^{1/p} \times \w_{\min}}{{\w^{*}}^{1/p}} \times \max\left\{ \lambda^{1/p}, \lambda^{\beta/p} \right\},
\end{align}
where $\hat{L}$ is the first $|\hat{L}|$-th element in sorted $L$ with respect to the number of nodes like $u$ with $\Psi_{\ell}(u) \geq k$ for $\ell \in \hat{L}$, and $\w_{\min}$ is the smallest layer weights that contributed to $C$ (i.e., removing it changes either $k$ or $\lambda$). 

\end{lemma}
\begin{proof}
By definition, each node $v \in C$ has at least $\Psi(u) \geq k$ in some layers with cumulative weights $\geq \lambda$, so based on the pigeonhole principle, there exists a layer $\ell'$ such that there are $\geq \frac{\lambda |C|}{\w^*}$ nodes like $u$ that each has $\Psi_{\ell'}(u) \geq k$. So we have:
$$\Omega_{\ell'}(|C|) \geq \w(\ell') \times \left( \frac{k \times \frac{\lambda |C|}{\w^*}}{|C|} \right)^{1/p} =  \w(\ell') \left( \frac{k \times \lambda}{\w^*}\right)^{1/p}.$$
Now, ignoring this layer, exploiting  the definition of $C$, and re-using the pigeonhole principle, we can conclude that there exists a layer $\ell''$ such that there are $\geq \frac{(\lambda - \w(\ell')) |C|}{\w^*}$ nodes like $u$ that each has $\Psi_{\ell''}(u) \geq k$. By iterating this process, we can simply conclude the Inequality~\ref{eq:density-firmcore}. Note that the last inequality is obtained from the first and last iterations of the above procedure.  
\end{proof}

\head{Case 1: $p \geq 1$}
Let $C^*_{SL}$ be the $(p+1)^{1/p}$ approx solution for $S^*_{SL}$ by~\cite{p-mean} (it exists when $p \geq 1$), and $\mu = \min \Delta_{\ell^*}(C^*_{SL})$. Since $C^*_{SL}$ is the optimal obtained solution, removing a node cannot increase its $p$-mean density (if increases, then we find a better approx solution as it is certainly produced in the algorithm). Therefore, it is simple to see that $\Omega_{\ell^*}(S^*_{SL})^p \leq \w(\ell^*)^p (p+1)\mu$. Based on the definition of $\mu$ and $\Delta$, there is a non-empty $(k^+, \lambda^+)$-GFirmCore that $k^+ \geq \mu$. So we have $k^+ \geq \frac{\Omega_{\ell^*}(S^*_{SL})^p}{\w(\ell^*)^p  (p+1)}$.

\begin{lemma}\label{lemma:3}
    $ \Omega_{\ell^*}(S^*_{SL}) {\w^*}^\beta \geq \rho(S^*)$.
\end{lemma}
\begin{proof}
$\Omega_{\ell^*}(S^*_{SL}) {\w^*}^\beta \geq \max_{\ell \in L} \Omega_{\ell}(S^*) {\w^*}^\beta
        \geq \max_{\hat{L} \subseteq L} \min_{\ell \in \hat{L}} \Omega_{\ell}(S^*) \left(\sum_{\ell' \in \hat{L}} \w(\ell')\right)^\beta.$
\end{proof}

\begin{theorem}[Approximation Algorithm for $p\geq 1$]\label{thm:p-approx}
    \begin{equation}
        \rho(C^*) \geq \frac{1}{(p+1)^{1/p}} \times \frac{\w_{\min} \times \max\left\{ {\lambda^+}^{1/p}, {\lambda^+}^{\beta/p} \right\}}{\w_{\max} {\w^{*}}^{\beta + 1/p}}  \times \rho(S^*),
    \end{equation}
\end{theorem}
\begin{proof}
    The proof of this theorem is based on Lemmas~\ref{density-firmcore} and \ref{lemma:3}, and the fact that $k^+ \geq \frac{\Omega_{\ell^*}(S^*_{SL})^p}{\w(\ell^*)^p  (p+1)}$.
\end{proof}
\noindent
Note that for the sake of simplicity, in the above theorem, we used Inequality~\ref{eq:density-firmcore}. For a tighter bound, one can use Inequality~\ref{eq:density-firmcore0} in Lemma~\ref{density-firmcore}. When $p = 1$ and $\w(.) = 1$, the approximation guarantee matches the approximation guarantee by \citet{FirmCore}, which is the best existing guarantee for this special case. Note that, our work is the first algorithm for the generalized $p$-mean MDS case.

\head{Case 2: $p \in [-\infty, 1]$} In this part, we show that our approx solution to $1$-mean MDS, can provide an approximation solution to $p$-mean MDS, when $p \in [-\infty, 1]$.

\begin{theorem}[Approximation Algorithm for $-\infty \leq p \leq 1$]\label{thm:2-approx}
    \begin{equation}
        \rho(C^*) \geq \frac{1}{(p+1)^{1/p}} \times \frac{\w_{\min} \times \max\left\{ {\lambda^+}^{1/p}, {\lambda^+}^{\beta/p} \right\}}{2 \times \w_{\max} {\w^{*}}^{\beta + 1/p}}  \times \rho(S^*),
    \end{equation}
\end{theorem}

\begin{proof}
     Let $S^{*^{(1)}}_{\text{SL}}$ be the optimal solution of $ \Omega_{\ell^*}(S^*_{\text{SL}})$ when $p = 1$. We know that $\min_{u \in S^{*^{(1)}}_{\text{SL}}} deg_{\ell^*}(u) \geq \frac{1}{2} \Omega_{\ell^*}(S^{*^{(1)}}_{\text{SL}}) = \frac{1}{2} \Omega^{(1)}_{\ell^*}(S^{*^{(1)}}_{\text{SL}})$ \underline{for $p = 1$}, since removing the node with the minimum degree cannot increase the density. On the other hand, as discussed by \citet{p-mean2}, $p$-mean function over the degree of nodes in a graph is monotone. Therefore, we have:
    \begin{align}
         \Omega_{\ell^*}(S^{*^{(1)}}_{\text{SL}}) \geq \min_{u \in S^*_{\text{SL}}} deg_{\ell^*}(u) \geq  \frac{1}{2} \Omega^{(1)}_{\ell^*}(S^*_{\text{SL}}) \geq \frac{1}{2} \Omega_{\ell^*}(S^*_{\text{SL}})
    \end{align}
    The last inequality comes from the monotonicity of $p$-mean function over the degree of nodes in a graph. Using Lemma~\ref{lemma:3} and Theorem~\ref{thm:p-approx}, we can simply show the above approximation guarantee. 
\end{proof}

\noindent
Note that, while empirically the value of $\alpha$ can affect the performance, theoretically its value cannot affect the approx guarantee as we only need $\alpha = 1$.

\begin{table*}[t!]
	\caption{
    Comparison of the solutions found by GFirmCore and the state-of-the-art FirmCore~\cite{FirmCore}. The superior performance of GFirmCore with different $p$ shows the importance of considering weights for different relation types.  
	}
	\label{tab:results}
	\centering
	\scalebox{0.67}{
 \hspace{-6ex}
 \begin{tabular}{l l l   |  l  | l l     l l l  l l l l l}
			\toprule
		&&   &\textbf{Dataset}&\textbf{Homo} & \textbf{Sacchcere} & \textbf{FAO  ~~~~ } & \textbf{Brain ~~~~}& \textbf{DBLP} & \textbf{Amazon} & \textbf{FFTwitter} & \textbf{Friendfeed} & \textbf{StackO} & \textbf{Google+}  \\
		&&   & $|V|$ & 18k & 6.5k & 214 & 190 & 513k & 410k & 155k & 510k & 2.6M & 28.9M\\
		&&   & $|E|$ & 153k & 247k & 319K & 934K & 1.0M & 8.1M & 13M & 18M & 47.9M & 1.19B\\
            && \textbf{Metric} & $|L|$ & 7 & 7 & 364 & 520 & 10 & 4 & 2 & 3 & 24 & 4 \\
			\midrule
            \midrule
			\multirow{13}{*}{\rotatebox[origin=c]{90}{GFirmCore} ~~} & & Edge  & $p = -\infty$ & 0.73 & 0.68& 0.45& 1.00 & 0.52 & 0.48& 0.74& 0.39 & 0.50 & 0.98\\
			&& Density& $p = - 1$& 0.73 & 0.49& 0.47& 1.00 & 0.39 & 0.48& 0.59& 0.36 & 0.53 & 0.56\\
			&& \multirow{2}{*}{$\frac{\sum_{\ell \in L}\w_\ell|E_\ell[S]|}{\w^* \times {\binom{|S|}{2}}}$} & $p = 0$ & 0.39 & 0.55 & 0.39& 0.92& 0.39 & 0.33 & 0.59& 0.78 & 0.46 & 0.73 \\
			&& & $p = 1$ &  0.58& 0.46& 0.47 & 0.90 & 0.39& 0.51& 0.59 & 0.48 & 0.53 & 0.84\\
			\cmidrule{2-14}
            \cmidrule{2-14}
			&& Multiplex  & $p = -\infty$ & 28.36 & 20.79& 1553.84 & 3941.55 & 77.46 & 41.89 & 111.42& 163.58 & 96.20 & 153.99\\
			&& Density~\cite{MLcore}& $p = - 1$& 30.17 & 19.53& 1559.25& 3941.55 & 81.17 & 42.01& 98.50 & 165.72  & 97.18 & 172.87\\
			& & & $p = 0$ &  28.49 & 31.26& 1674.41& 7180.09 & 82.46 & 40.51& 98.73& 183.76 & 99.03 & 148.16\\
			&& & $p = 1$& 31.14 & 28.59 & 1854.07 & 7935.29&  82.91 & 61.38 & 99.26& 216.74 & 118.33 & 173.81\\
						\cmidrule{2-14}
                        \cmidrule{2-14}
			&& Runtime ($s$) & $p = -\infty$ &  38 & 96& 7199& 9207 & 930 & 992& 894& 4375 & 23698 & 71148\\
            && & $p = - 1$& 43 &  101& 7418& 9491 & 1061 & 1206& 1089& 4810 & 26056 & 74703\\
			&& & $p = 0 $ & 39  & 113& 7407& 9462 & 1128 & 1135& 1103& 4729 & 26114 & 74669\\
			&& & $p = 1$& 48 &  105& 7369 & 9503 & 1076 & 1160& 1057 & 4788 & 25671 & 74893\\
			\bottomrule
            \bottomrule
            \multirow{13}{*}{\rotatebox[origin=c]{90}{FirmCore}} & & Edge  & $p = -\infty$ & 0.69 & 0.61 & 0.45 & 0.92 & 0.44 & 0.37 & 0.60 & 0.42  & 0.46  & 0.74\\
			&& Density& $p = - 1$& 0.58 & 0.61 & 0.45 & 0.92 & 0.35 & 0.33 & 0.52 & 0.38  & 0.49  & 0.70\\
			&& \multirow{2}{*}{$\frac{\sum_{\ell \in L}\w_\ell|E_\ell[S]|}{\w^* \times {\binom{|S|}{2}}}$} & $p = 0$ & 0.32 & 0.61 & 0.39 & 0.92 & 0.35 & 0.31 &   0.52  & 0.36 & 0.41 & 0.52 \\
			&& & $p = 1$ & 0.47 & 0.42 & 0.35 & 0.78 & 0.41 & 0.42 & 0.52 & 0.36  & 0.45  & 0.52\\
			\cmidrule{2-14}
            \cmidrule{2-14}
			&& Multiplex  & $p = -\infty$ & 27.85 & 22.91 & 1553.84 & 6997.12 & 75.19 & 39.28 & 98.46 & 167.19  & 98.51  & 162.43\\
			&& Density~\cite{MLcore}& $p = - 1$& 28.14 & 23.69 & 1598.66 & 7034.50 & 75.83 & 39.15 & 98.03 & 167.56  & 100.03  & 163.88\\
			&& & $p = 0$ & 28.53 & 25.82 & 1659.41 & 7180.09 & 76.11 & 39.64 & 99.12 & 168.44  & 100.98  & 162.07\\
			&& & $p = 1$ & 29.74 & 25.87 & 1673.18 & 7163.89 & 78.91 & 43.52 & 100.24 & 170.87  & 107.09  & 164.81\\
						\cmidrule{2-14}
                        \cmidrule{2-14}
			&& Runtime ($s$) & $p = -\infty$ & 19 & 36 & 2403 & 3169 & 322 & 348 & 297 & 799  & 6951  & 34814\\
            && & $p = - 1$& 21 & 37 & 2964 & 3613 & 438 & 489 & 386 & 841  & 8116 & 35726 \\
			&& & $p = 0 $ & 20 & 46 & 2954 & 3486 & 447 & 467 & 394 & 835  & 8170  & 35482 \\
			&& & $p = 1$& 20 & 41 & 2454 & 3273 & 362 & 394 & 359 & 891  & 8053  & 36027 \\
   \bottomrule
	\end{tabular}
 }\vspace{-3ex}
\end{table*} 

\section{Experiments}\label{sec:experiments}

\head{Setup}
 Designed algorithms and baselines are implemented in Python (compiled by Cython). All experiments are performed on a Linux machine with Intel Xeon 2.6 GHz CPU and 128 GB RAM.

\head{Datasets}
In our experiments, we use 10 real-world datasets~\cite{FirmCore, FirmTruss, CS-MLGCN, MLcore, Friendfeed, Higgs, amazon_datset, Google+, FAO, homo, Twitter_datasets} whose domains cover social, genetic, co-authorship, financial, and co-purchasing networks. The main characteristics are summarized in \autoref{tab:results}. We use an unsupervised learning method to learn the importance of each layer~\cite{anomuly} and treat them as layer~weights.

\head{Results}
Table~\ref{tab:results} reports the average edge density and multiplex density for different values of $p$. Based on these results, our definition of density can find different and meaningful dense structures. Also, it is notable that the effect of $p$ on the performance depends on the datasets, which again shows the importance of the flexibility that our formulation can provide. GFirmCore in all datasets finds a densest structure that is denser than the found solution by FirmCore, which shows the significance of considering weights for different layers. 

\noindent
Since there is no algorithm for exactly finding the multiplex densest subgraph, we generate two synthetic datasets, S1 and S2, both with $|V| = 100$, $|E| = 10000$, $|L| = 4$. We use the same approach as real-world datasets to obtain layer weights. We also inject the densest subgraph via clique density to S1 and average degree density to S2. Figure~\ref{fig:approx} reports the ratio of the found solution and the optimal solution obtained by our algorithms ($p = 1, 2, 3$) and baselines FirmCore~\cite{FirmCore} and ML $\mathbf{k}$-core~\cite{MLcore}. Our algorithms outperform both baselines in both datasets and all values of $p$ including $p = 1$, which they are designed for. This result shows the importance of handling different importance for~different~layers.

\noindent
Figure~\ref{fig:time} shows the running time of our algorithms and baselines. While our algorithms are much faster than ML $\mathbf{k}$-core~\cite{MLcore}, FirmCore is more efficient than our algorithms. The main reason is that FirmCore does not consider different weights and as we discussed in \S\ref{sec:method}, this relaxation can change the complexity of the decomposition (GFirmCore is NP-hard while FirmCore is polynomial). It is notable that our algorithms are scalable to graphs with~billions~of~edges.

\begin{figure}
\centering
\hspace{-6ex}
\begin{minipage}{.48\textwidth}
  \centering
    \includegraphics[width=1.1\linewidth]{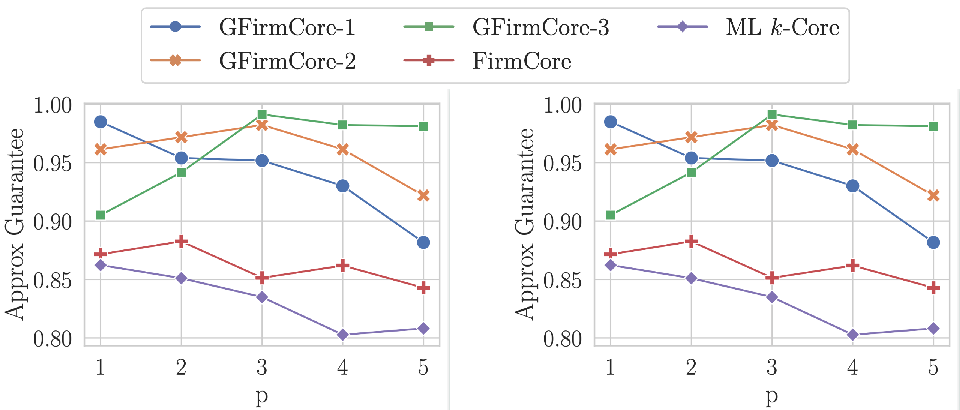}
    \vspace{
    -3ex
    }
  \caption{The quality of found solution by GFirmCore and baselines. (Left) S1, (Right) S2 datasets.}
 \label{fig:approx}
\end{minipage}
\hspace{7ex}
\begin{minipage}{.48\textwidth}
  \centering
    \includegraphics[width=1.1\linewidth]{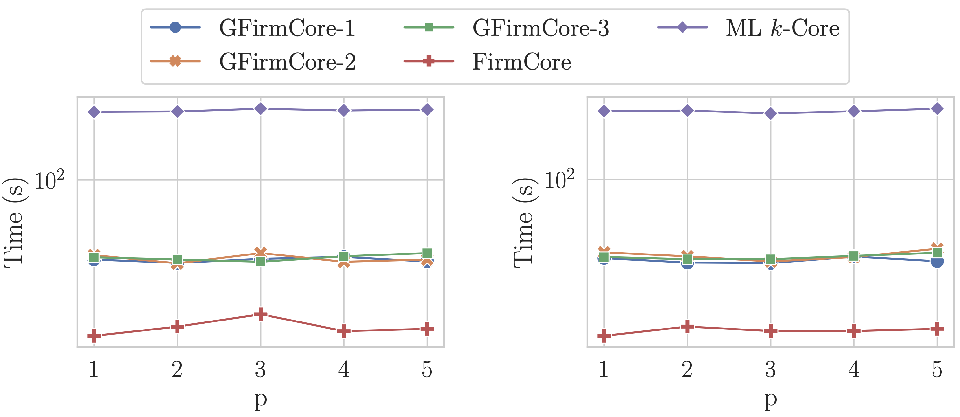}
    \vspace{
    -3ex
    }
  \caption{The running time of GFirmCore and baselines. (Left) S1, (Right) S2 datasets.}
  \label{fig:time}
\end{minipage}
\vspace{-2ex}
\end{figure}

\begin{figure}
    \centering
    \includegraphics[width=0.90\linewidth]{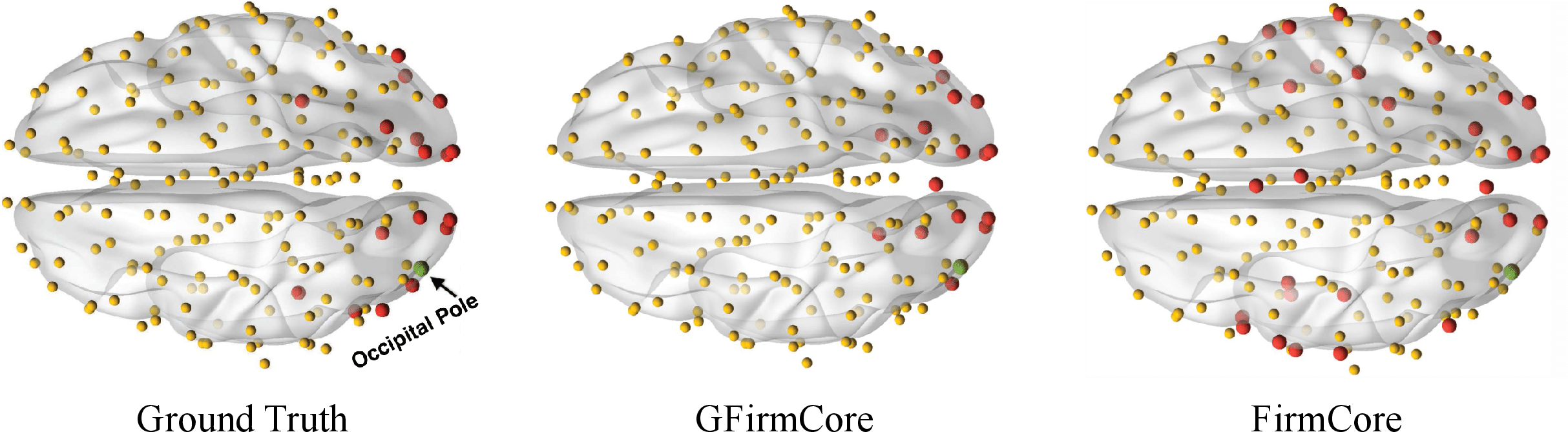}
    \vspace{
    -2ex
    }
    \caption{The running time of GFirmCore and baselines. (Left) S1, (Right) S2 datasets.}
    \label{fig:case-study}
    \vspace{-2ex}
\end{figure}

\head{Case study: Brain Networks}
Detecting and monitoring functional systems in the human brain is a primary task in neuroscience. Brain Networks obtained from fMRI, are graph representations of the brain, where each node is a brain region and two nodes are connected if there is a high correlation between their functionality. However, the brain network generated from an individual can be noisy and incomplete. Using brain networks from many individuals can help to identify functional systems more accurately. A dense subgraph in a multiplex brain network, where each layer is the brain network of an individual, can be interpreted as a functional system in the brain. Figure~\ref{fig:case-study} shows the densest subgraph including the occipital pole found by FirmCore and GFirmCore as well as the ground-truth functional system of the occipital pole (i.e., visual processing). The densest subgraph found by GFirmCore is more similar to ground truth than FirmCore. The main reason is that the brain network generated from an individual can be noisy/incomplete and FirmCore treats all layers the same.

\vspace{-2ex}
\section{Conclusion}
In this paper, we propose and study a novel extended notion of core in layer-weighted multiplex networks, GFirmCore, where each layer has a weight that indicates the importance/significance of the layer. We show that theoretically this problem is more challenging than its layer-unweighted counterpart and is NP-hard. We further extend the notion of multiplex density to layer-weighted multiplex networks. For the sake of unifying existing density measures, we propose a new family of densest subgraph objectives, parameterized by a single parameter $p$ that controls the importance of larger/smaller degrees in the subgraph. Using our GFirmCore, we propose the first polynomial approximation algorithm that provides approximation guarantee in the general case of $p$-mean densest subgraph problem. Our experimental results, show the efficiency and effectiveness of our algorithms and the significance of considering different weights for the layers in multiplex networks.

% %
% % ---- Bibliography ----
% %
% \begin{thebibliography}{6}
% %

% \bibitem {bib1}
% Smith, T.F., Waterman, M.S.: Identification of common molecular subsequences.
% J. Mol. Biol. 147, 195?197 (1981). \url{doi:10.1016/0022-2836(81)90087-5}

% \bibitem {may:ehr:stein}
% May, P., Ehrlich, H.-C., Steinke, T.: ZIB structure prediction pipeline:
% composing a complex biological workflow through web services.
% In: Nagel, W.E., Walter, W.V., Lehner, W. (eds.) Euro-Par 2006.
% LNCS, vol. 4128, pp. 1148?1158. Springer, Heidelberg (2006).
% \url{doi:10.1007/11823285_121}

% \bibitem {fost:kes}
% Foster, I., Kesselman, C.: The Grid: Blueprint for a New Computing Infrastructure.
% Morgan Kaufmann, San Francisco (1999)

% \bibitem {czaj:fitz}
% Czajkowski, K., Fitzgerald, S., Foster, I., Kesselman, C.: Grid information services
% for distributed resource sharing. In: 10th IEEE International Symposium
% on High Performance Distributed Computing, pp. 181?184. IEEE Press, New York (2001).
% \url{doi: 10.1109/HPDC.2001.945188}

% \bibitem {fo:kes:nic:tue}
% Foster, I., Kesselman, C., Nick, J., Tuecke, S.: The physiology of the grid: an open grid services architecture for distributed systems integration. Technical report, Global Grid
% Forum (2002)

% \bibitem {onlyurl}
% National Center for Biotechnology Information. \url{http://www.ncbi.nlm.nih.gov}

% \end{thebibliography}

% ---- OR ------- Uncomment this section to use bibtex
\bibliographystyle{plainnat} % We choose the "plain" reference style
\bibliography{main} % Entries are in the refs.bib file
\end{document}